\documentclass{article}
\usepackage{amsmath}
\usepackage{amsthm}
\usepackage{amssymb}
\usepackage{graphicx}
\usepackage{natbib}

\theoremstyle{plain}
\newtheorem{thm}{Theorem}
\newtheorem{lem}{Lemma}
\newtheorem{assumption}{Assumption}

\newtheorem{manualtheoreminner}{Lemma}
\newenvironment{manualtheorem}[1]{%
  \renewcommand\themanualtheoreminner{#1}%
  \manualtheoreminner
}{\endmanualtheoreminner}

\theoremstyle{definition}
\newtheorem{defn}{Definition}

\title{Modeling the effects of dynamic range compression on signals in noise}
\author{Ryan M. Corey and Andrew C. Singer\\ University of Illinois at Urbana-Champaign%
\thanks{This research was supported by the National Science Foundation under
Grant No. 1919257 and by an appointment to the Intelligence Community
Postdoctoral Research Fellowship Program at the University of Illinois
at Urbana-Champaign, administered by Oak Ridge Institute for Science
and Education through an interagency agreement between the U.S. Department
of Energy and the Office of the Director of National Intelligence.}
}
\vspace{-8ex}\date{}

\begin{document}

\maketitle

\begin{abstract}
Hearing aids use dynamic range compression (DRC),
a form of automatic gain control, to make quiet sounds louder and
loud sounds quieter. Compression can improve listening comfort,
but it can also cause distortion in noisy environments. It has been
widely reported that DRC performs poorly in noise, but there has been
little mathematical analysis of these distortion effects. This work
introduces a mathematical model to study the behavior of DRC in noise.
Using statistical assumptions about the signal envelopes, we define
an effective compression function that models the compression applied
to one signal in the presence of another. This framework is used to
prove results about DRC that have been previously observed experimentally:
that when DRC is applied to a mixture of signals, uncorrelated signal
envelopes become negatively correlated; that the effective compression
applied to each sound in a mixture is weaker than it would have been
for the signal alone; and that compression can reduce the long-term
signal-to-noise ratio in certain conditions. These theoretical results
are supported by software experiments using recorded speech signals.
\end{abstract}

\section{Introduction }

Hearing aids often perform poorly where people with hearing loss need help most: in noisy
environments with many competing sound sources. One challenge for
hearing aids in noise is a nonlinear processing technique known as
\emph{dynamic range compression} (DRC), which makes quiet sounds louder
and loud sounds quieter \citep{souza2002effects,allen2003amplitude,kates2005principles}.
Compression is used in all modern hearing aids, but it can cause distortion
when applied to multiple overlapping signals. For example, a sudden
loud noise can reduce the gain applied to speech sounds. This effect
is well documented empirically, but has been little studied mathematically.
To better understand the performance of DRC in noisy environments,
this work applies tools from signal processing theory to model the
effects of DRC on mixtures of multiple signals.

The auditory systems of people with hearing loss have reduced dynamic
range compared to those of normal-hearing listeners: Quiet sounds
need to be amplified in order to be audible, but loud sounds can cause
discomfort. Hearing aids with DRC apply level-dependent amplification so that
the output signal has a smaller dynamic range than the input signal.
A typical DRC system is shown in Fig. \ref{fig:DRC-diagram}. An envelope
detector tracks the level of the input signal over time in one or
more frequency bands while a compression function adjusts the
amplification to keep the signal level
within a comfortable range. Both the envelope detector and the compression
function are nonlinear processes, so changes in one signal can affect
the processing applied to other signals.

This interaction between signals can be difficult to measure, but
hearing researchers have found three quantifiable effects of compression
in noise. First, as one signal grows louder, it reduces the gain applied
to the mixture, thereby lowering the level of the other signal(s)
in the output. This effect has been called co-modulation \citep{stone2004side}
or across-source modulation \citep{stone2007quantifying}. Second,
noise can reduce the effect of a compressor, especially at low signal-to-noise
ratios (SNR) \citep{souza2006measuring}. If one signal has much higher
level than another, the DRC system will apply gain based on the stronger
signal and will have little effect on the dynamic range of the weaker
signal. Finally, at high SNR, compressors tend to amplify low-level
noise more strongly than the higher-level signal of interest, which
can reduce the average SNR \citep{souza2006measuring,rhebergen2009dynamic,alexander2015effects}.
This effect has been observed in commercial hearing aids and shown to impair
speech comprehension \citep{naylor2009long,miller2017output}.

The adverse effects of noise on DRC systems have been well documented
empirically, but the problem has received
little mathematical analysis, which could provide more insight than
empirical evidence alone. This work applies signal processing research
methods to the DRC distortion problem: First, we make simplifying
assumptions to develop a tractable mathematical model of a complex
system. Next, we use that model to prove theorems that explain the
behavior of the system. Finally, we validate those theoretical results
using realistic experiments. 

\begin{figure}
\begin{centering}
\includegraphics{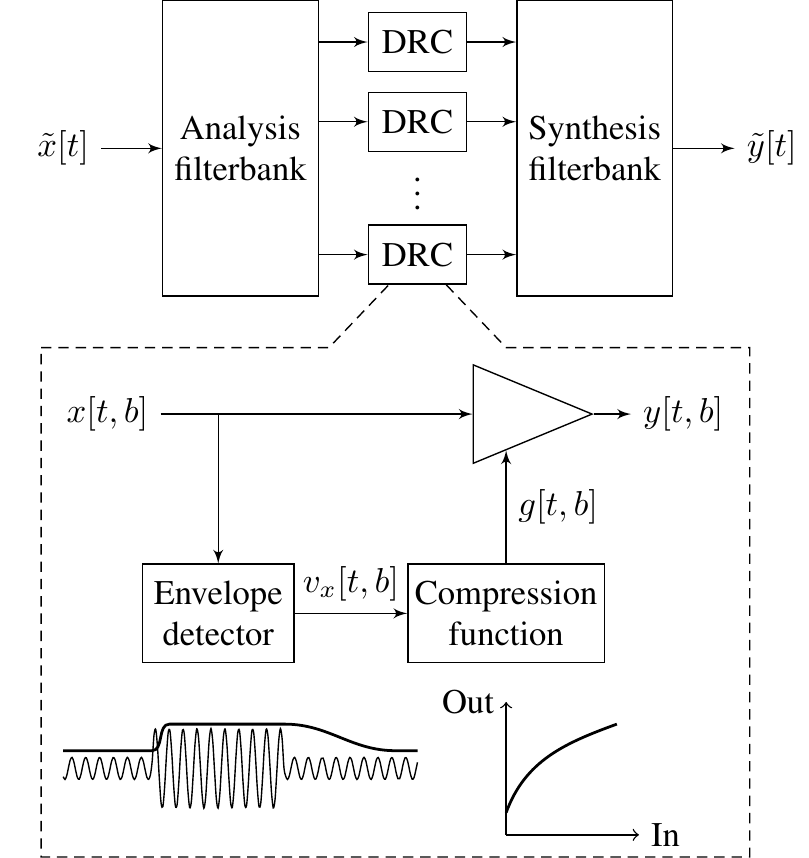}
\par\end{centering}
\caption{\label{fig:DRC-diagram}A typical DRC system performs automatic gain
control in each of several frequency bands or channels.}
\end{figure}

Compression systems are difficult to analyze because of the complex
interactions between the envelope detector and compression function,
both of which are nonlinear. Using the simplifying assumption that
envelopes are additive in signal mixtures, we can separate the effects
of the envelope detector from those of the compression function. To characterize
the interaction between signals in a mixture, we introduce the effective
compression function (ECF), which relates the input and output levels
of one signal in the presence of another. The ECF is used to explain
the three effects described above: that compression induces negative
correlation between signal envelopes (Section \ref{sec:across_source}),
that noise reduces the effect of compression (Section \ref{sec:effective_compression}),
and that compression can reduce average SNR in certain conditions
(Section \ref{sec:snr}). Each section includes a theorem
about the effect and simulation experiments that illustrate it. Wherever
possible, these results are compared to those reported in the hearing literature. 

\section{\label{sec:drc_overview}Dynamic Range Compression of a Single Signal}

Because most modern hearing aids are digital, we formulate the DRC
system in discrete time. Let the sequence $\tilde{x}[t]$ be a sampled
audio signal at the input of the DRC system, where $t$ is the sample
index. Let $\tilde{y}[t]$ be the output of the system.

\subsection{Filterbank decomposition}

In hearing aids, DRC is often performed separately in several frequency
bands. A filterbank splits the signal into $B$ channels corresponding
to different bands, which may be uniform or nonuniform and may or
may not overlap. Let $x[t,b]$ and $y[t,b]$ be the filterbank representations
of $\tilde{x}[t]$ and $\tilde{y}[t]$, respectively, in channels
$b=1,\dots,B$. 

The number and structure of channels are known to have significant
effects on the performance of DRC systems in noise \citep{naylor2009long,alexander2015effects,rallapalli2019effects}.
The levels of signals in narrower bands fluctuate more rapidly than
the levels of wideband signals, and so systems with more channels
tend to compress signals more strongly and to produce greater distortion
effects. The experiments in this work use a Mel-spaced filterbank
with 6 bands, which are roughly uniformly spaced at lower frequencies
and exponentially spaced at higher frequencies.

\subsection{Envelope detection}

The gain applied by DRC is calculated from the signal envelope,
which tracks the signal level over time. Signal level is typically defined in
terms of either magnitude ($|x|$) or power ($x^{2}$); this work
uses power. Let the nonnegative signal
$v_{x}[t,b]$ be the envelope of the input signal $x[t,b]$ at time
index $t$ and channel $b$. In the theoretical analysis presented here, the envelope
is an abstract property of a signal. For example, it can be thought
of as the time-varying statistical variance of a random process of
which the signal $x[t,b]$ is a realization. In real DRC systems,
the envelope is estimated from the observed signal.

Because hearing aids must process signals with imperceptible delay,
a practical envelope detector estimates signal level based on a moving
average of past samples. Most DRC systems respond faster to increases
in signal level (attack mode) than to decreases in signal level (release
mode). Short attack times, typically a few milliseconds, help to suppress
sudden loud sounds. Longer release times, from tens to hundreds of milliseconds, ensure that gain
is not increased too quickly during brief pauses \citep{jenstad2005quantifying}.

There are many ways of implementing an envelope detector \citep{giannoulis2012digital}.
A representative detector, which is used in the experiments throughout
this work, is the nonlinear recursive filter \citep{kates2008digital}
\begin{equation}
v_{x}[t,b]=\begin{cases}
\beta_{\text{a}}v_{x}[t\!-\!1,b]\!+\!(1\!-\!\beta_{\mathrm{a}})\left|x[t,b]\right|^{2},\\
 & \hspace{-2cm}\text{if }\left|x[t,b]\right|^{2}\ge v_{x}[t\!-\!1,b]\\
\beta_{\text{r}}v_{x}[t\!-\!1,b]\!+\!(1\!-\!\beta_{\mathrm{r}})\left|x[t,b]\right|^{2}, & \text{otherwise},
\end{cases}\label{eq:env_det}
\end{equation}
for $b=1,\dots,B,$ where $\beta_{\mathrm{a}}$ and $\beta_{\mathrm{r}}$
are constants that determine the attack and release times.

Because envelope detection is a nonlinear process, it contributes
to the distortion effects of DRC systems. The theorems in this work do not depend directly on the choice
of attack and release time, but these parameters do affect the distribution
of envelope samples: A slowly-changing envelope detector will measure
a narrower dynamic range than a fast-changing detector for the same
signal. Many distortion effects are therefore more severe for fast-acting
than for slow-acting compression \citep{jenstad2005quantifying,naylor2009long,alexander2015effects,reinhart2017effects,alexander2017acoustic}. 

\subsection{Compression function}

\begin{figure}
\begin{centering}
\includegraphics{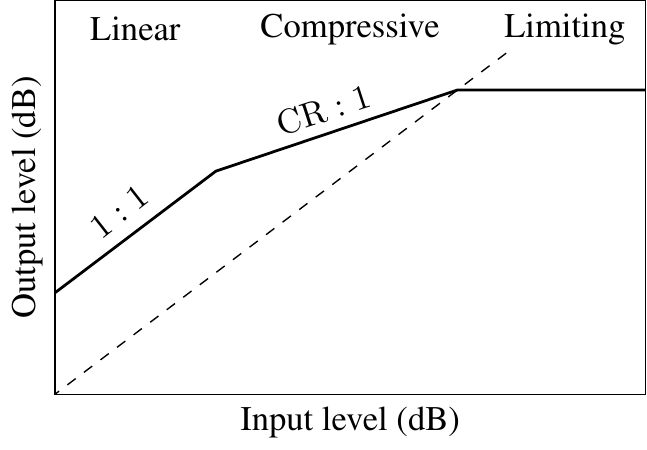}
\par\end{centering}
\caption{\label{fig:Compression-functions}A compression function $\mathcal{C}_{b}$,
shown here on a logarithmic scale, maps input levels
to output levels.}
\end{figure}

A compression function $\mathcal{C}_{b}$ determines the mapping between
input and output level in each channel:
\begin{equation}
v_{y}[t,b]=\mathcal{C}_{b}\left(v_{x}[t,b]\right),\quad b=1,\dots,B,
\end{equation}
where $v_{y}[t,b]$ is the target output level. The amplification
applied in each channel is then 
\begin{equation}
g[t,b]=\sqrt{\frac{v_{y}[t,b]}{v_{x}[t,b]}},\quad b=1,\dots,B,
\end{equation}
so that the output is 
\begin{equation}
y[t,b]=g[t,b]x[t,b],\quad b=1,\dots,B.
\end{equation}

Note that the target output level $v_{y}[t,b]$ is not necessarily
equal to the measured envelope of $y[t,b]$ because
the envelope is a nonlinear moving average of present and
past levels. Since this distinction is not important to our analysis, similar notation will be used for both.

Although compression functions are defined here in terms of input
and output level (i.e., power), they are often visualized and described on a logarithmic
scale, such as in dB SPL (sound pressure level). A typical ``knee-shaped''
compression function is shown in Fig. \ref{fig:Compression-functions}:
It features a linear region in which gain is constant, a compressive
region where the output level increases by less than the input level,
and and a limiting region that prevents the output from exceeding
a maximum safe level. 

The strength of compression can be characterized by the compression
ratio (CR), which is the inverse of the slope of the compression function on a log-log scale.
For example, in a 3:1 compression system, the output level increases
by 1 dB for every 3 dB increase in the input level. For a DRC system
with constant compression ratio $\mathrm{CR}$, the compression function
is given by the power-law relationship
\begin{equation}
\mathcal{C}_{b}(v)=g_{0}^{2}[b]v^{1/\mathrm{CR}},
\end{equation}
where $g_{0}^{2}[b]$ is a constant gain factor. Thus, for a 3:1 compressor,
the output level is proportional the cube root of the input level.
In limiters, $\mathcal{C}_{b}(v)$ is constant and so the compression
ratio is infinite.

While most compressors reported in the literature use
some combination of linear, power-law, and limiting compression functions,
many others are possible. For example, cascaded feedback systems can
be used to design smoothly curved compression functions with roughly
logarithmic shapes \citep{lyon2017human}. In digital systems, $\mathcal{C}_{b}$
can be arbitrary. To make our analysis as general as possible, we
allow the compression function to be any mapping between nonnegative
numbers such that the output level grows no faster than the input
level.
\begin{defn}
\label{def:cf}A function $\mathcal{C}_{b}(v)$ is a \emph{compression
function} if it is concave, nonnegative, and nondecreasing for all
$v>0$.
\end{defn}

To describe how much a compression function reduces the dynamic range
of a signal, we could compute its compression ratio. Because the compression
ratio can be infinite, however, it is more convenient to work with
its inverse, the compression slope.
\begin{defn}
\label{def:cs}For all points $v$ at which a compression function
$\mathcal{C}_{b}(v)$ is differentiable, the \emph{compression slope}
$\mathrm{CS}_{b}(v)$ is the slope of $\mathcal{C}_{b}(v)$ on a log-log
scale:
\begin{align}
\mathrm{CS}_{b}(v) & =\frac{\mathrm{d}}{\mathrm{d}u}\ln\mathcal{C}_{b}\left(e^{u}\right)\mid_{u=\ln v}\\
 & =\frac{\mathcal{C}'_{b}(v)}{\mathcal{C}_{b}(v)}v.
\end{align}
\end{defn}

For example, if $\mathcal{C}_{b}(v)=g_{0}^{2}[b]v^{\alpha}$, then
$\mathrm{CS}_{b}(v)=\alpha$ for all $v$. The smaller the compression
slope, the more the dynamic range of the signal is reduced.

\section{\label{sec:mixtures}Modeling Compression of Sound Mixtures}

\begin{figure}
\begin{centering}
\includegraphics{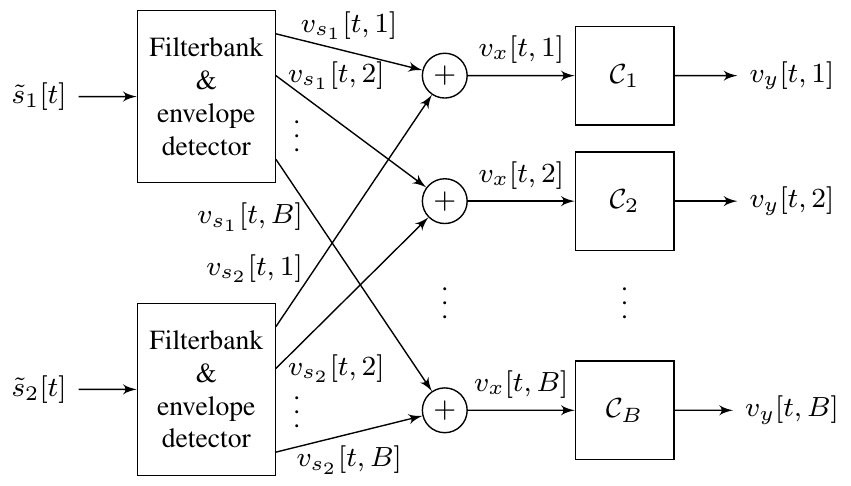}
\par\end{centering}
\caption{\label{fig:A-simplified-model}A simplified model separates the effects
of the filterbank and envelope detector from those of the compression
functions $\mathcal{C}_{b}$. The former is assumed to act independently
across component signals, while the latter act independently across
time and channels.}

\end{figure}

Hearing aids are often used in noisy environments with multiple sound
sources. Because DRC is a nonlinear process, these signals interact
and cause distortion. This distortion is especially difficult to analyze
because DRC involves two nonlinear operations: envelope detection
and level-dependent amplification. To create a tractable model, we
make a simplifying assumption about the signal envelopes that allows
us to separate the effects of these nonlinarities, as shown in Fig.
\ref{fig:A-simplified-model}. Under this model, the filterbank and
envelope detector determine the relationship between input signals
and envelope values; they are modeled as acting independently on each
source signal. The compression function determines the output levels
from these envelopes; it acts independently at each time index and
within each channel.

\subsection{Envelope model}

Suppose that the input to the system is $\tilde{x}[t]=\tilde{s}_{1}[t]+\tilde{s}_{2}[t]$,
where $\tilde{s}_{1}[t]$ and $\tilde{s}_{2}[t]$ are two discrete-time
signals. Because a filterbank is a linear system, the filterbank representation
of the input is 
\begin{equation}
x[t,b]=s_{1}[t,b]+s_{2}[t,b],\quad b=1,\dots,B,\label{eq:band_mixture}
\end{equation}
where $s_{1}[t,b]$ and $s_{2}[t,b]$ are the filterbank representations
of $\tilde{s}_{1}[t]$ and $\tilde{s}_{2}[t]$, respectively.

Because envelope detection is a nonlinear process, the additivity
property of Eq. (\ref{eq:band_mixture}) does not hold in general
for the signal envelopes measured by practical envelope detectors.
However, to simplify our analysis, the signal envelopes can be \emph{modeled}
as obeying additivity. 
\begin{assumption}
\label{assu:envelopes}The envelopes $v_{s_{1}}[t,b]$, $v_{s_{2}}[t,b]$,
and $v_{x}[t,b]$ of $s_{1}[t,b]$, $s_{2}[t,b]$, and $x[t,b]$ satisfy
\begin{equation}
v_{x}[t,b]=v_{s_{1}}[t,b]+v_{s_{2}}[t,b],\quad b=1,\dots,B.
\end{equation}
\end{assumption}

To justify this assumption, suppose that $s_{1}[t,b]$ and $s_{2}[t,b]$
are sample functions of zero-mean random processes that are 
uncorrelated with each other. Then $\mathbb{E}\left[|x[t,b]|^{2}\right]=\mathbb{E}\left[|s_{1}[t,b]|^{2}\right]+\mathbb{E}\left[|s_{2}[t,b]|^{2}\right]$,
where $\mathbb{E}$ denotes the statistical expectation. If $v_{x}[t,b]$
were any linear transformation of the sequence $\mathbb{E}\left[|x[t,b]|^{2}\right]$,
then the envelopes would satisfy Assumption \ref{assu:envelopes}.
This statistical model is useful because the instantaneous level $\left|x[t,b]\right|^{2}$
from Eq. (\ref{eq:env_det}) can be thought of as an estimator of
the variance $\mathbb{E}\left[|x[t,b]|^{2}\right]$. If $\beta_{\mathrm{a}}=\beta_{\mathrm{r}}$,
then the recursive filter would be a linear transformation of this
estimator. This model does not reflect the peak-tracking behavior
of envelope detectors with $\beta_{\mathrm{a}}\ne\beta_{\mathrm{r}}$.
However, the simulation experiments in this work do include peak tracking.

\subsection{Output model}

Care is also required in analyzing the components of the output of
a nonlinear system. Let $\tilde{y}[t]=\tilde{r}_{1}[t]+\tilde{r}_{2}[t]$,
where $\tilde{r}_{1}[t]$ is the component of the output corresponding
to $\tilde{s}_{1}[t]$ and $\tilde{r}_{2}[t]$ is the component corresponding
to $\tilde{s}_{2}[t]$. For systems with the additivity property,
like linear filters, these components can be calculated by applying
the same system to $\tilde{s}_{1}$ and $\tilde{s}_{2}$. For nonlinear
systems like DRC, each component of the output depends on both components
of the input, so it can be difficult to meaningfully decompose the
output into distinct components. For example, in certain musical
genres, DRC is used to produce deliberately strong distortion so that
the original signals are barely recognizable in the output. In hearing
aids, however, the effects of DRC should be perceptually transparent:
The distortion should be subtle enough that an approximate notion
of additivity can apply.

In this work, the output components are determined by calculating
the level-dependent amplification sequence $g[t,b]$ based on the
mixture $x[t,b]$, then applying it to each component:
\begin{align}
y[t,b] & =g[t,b]x[t,b]\\
 & =g[t,b]\left(s_{1}[t,b]+s_{2}[t,b]\right)\\
 & =\underbrace{g[t,b]s_{1}[t,b]}_{r_{1}[t,b]}+\underbrace{g[t,b]s_{2}[t,b]}_{r_{2}[t,b]}
\end{align}
for all time indices $t$ and channels $b=1,\dots,B$. This decomposition
is used in the mathematical analysis below. Similarly, in the software
simulations presented here, the separate input signals are stored
in memory alongside their mixture and the amplification sequence is
applied separately to each, allowing the output components to be computed
exactly. In laboratory experiments with real hearing aids, many researchers
use a linearization technique known as phase inversion \citep{hagerman2004method}
to estimate the output components due to each source signal.

\subsection{Effective compression function}

\begin{figure}
\begin{centering}
\includegraphics{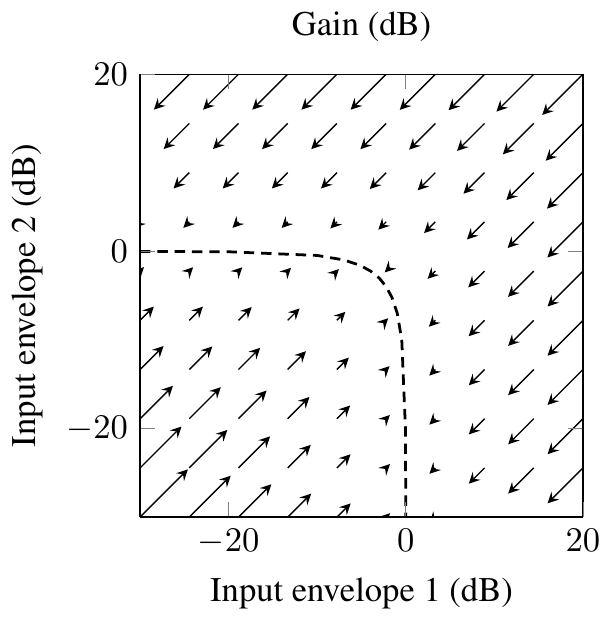}
\par\end{centering}
\caption{\label{fig:quiver}Gain applied to a mixture signal
as a function of the signal envelopes $v_{s_{1}}[t,b]$ and $v_{s_{2}}[t,b]$
for $\mathcal{C}_{b}(v)=v^{1/3}$ under Assumption \ref{assu:envelopes}.
The length of the arrows is proportional to the power gain $g^2[t,b]$ in dB and
the dashed curve shows the equilibrium mixture level $\mathcal{C}_{b}(v_{s_{1}}+v_{s_{2}})=v_{s_{1}}+v_{s_{2}}$.
The axes are scaled so that this equilibrium level is 0 dB.}

\end{figure}

The additive models for the envelope and output signals, while imperfect,
allow us to study the dominant source of nonlinearity in a DRC system:
the compression function. Although the signals $s_{1}[t,b]$ and $s_{2}[t,b]$
may have different levels, the amplification $g[t,b]$ applied to
both of them is the same and is computed from the overall level of
the input signal:
\begin{align}
g[t,b] & =\sqrt{\frac{\mathcal{C}_{b}(v_{x}[t,b])}{v_{x}[t,b]}}.
\end{align}
Under Assumption \ref{assu:envelopes}, the amplification is 
\begin{equation}
g[t,b]=\sqrt{\frac{\mathcal{C}_{b}(v_{s_{1}}[t,b]+v_{s_{2}}[t,b])}{v_{s_{1}}[t,b]+v_{s_{2}}[t,b]}},
\end{equation}
resulting in the output levels 
\begin{align}
v_{r_{1}}[t,b] & =\frac{\mathcal{C}_{b}(v_{s_{1}}[t,b]+v_{s_{2}}[t,b])}{v_{s_{1}}[t,b]+v_{s_{2}}[t,b]}v_{s_{1}}[t,b]\label{eq:env_out_1}\\
v_{r_{2}}[t,b] & =\frac{\mathcal{C}_{b}(v_{s_{1}}[t,b]+v_{s_{2}}[t,b])}{v_{s_{1}}[t,b]+v_{s_{2}}[t,b]}v_{s_{2}}[t,b],\label{eq:env_out_2}
\end{align}
for channels $b=1,\dots,B$. The gain and therefore the output levels
are functions of both input signal levels, as illustrated in Fig.
\ref{fig:quiver}. The gain applied to $s_{1}[t,b]$ in the presence
of $s_{2}[t,b]$ is weaker than it would have been for $s_{1}[t,b]$
alone. To characterize this effect, we can define an effective compression
function that relates the input and output levels of one signal in
the presence of another.
\begin{defn}
\label{def:ecf}The \emph{effective compression function }(ECF) $\hat{\mathcal{C}}_{b}(v_{1}|v_{2})$
applied to a signal with level $v_{1}>0$ in the presence of a signal
with level $v_{2}\ge0$ is given by
\begin{equation}
\hat{\mathcal{C}}_{b}(v_{1}|v_{2})=\frac{\mathcal{C}_{b}(v_{1}+v_{2})}{v_{1}+v_{2}}v_{1},
\end{equation}
where $\mathcal{C}_{b}(v)$ is the compression function applied to
the mixture level $v_{1}+v_{2}$.
\end{defn}
Using this definition, Eqs. (\ref{eq:env_out_1}) and (\ref{eq:env_out_2})
become
\begin{align}
v_{r_{1}}[t,b] & =\hat{\mathcal{C}}_{b}(v_{s_{1}}[t,b]|v_{s_{2}}[t,b])\\
v_{r_{2}}[t,b] & =\hat{\mathcal{C}}_{b}(v_{s_{2}}[t,b]|v_{s_{1}}[t,b])
\end{align}
for $b=1,\dots,B$. The ECF expresses the dependence between the levels
of the two signal components. It can be used to mathematically characterize
the distortion introduced by DRC systems in noisy environments, including
the across-source modulation effect, the effective compression ratio,
and the signal-to-noise ratio.

\section{\label{sec:across_source}Across-Source Modulation Distortion}

Dynamic range compression creates distortion in mixtures because the
presence of one signal alters the gain applied to another signal.
It has been observed experimentally \citep{stone2004side,stone2007quantifying,stone2008effects,alexander2015effects}
that when two signals are mixed together and passed through a compressor,
their output envelopes become negatively correlated: As one sound
becomes louder, the other sound becomes quieter. The across-source
modulation coeffient, a measure of this negative correlation, was
found to be correlated with reduced speech intelligibility \citep{stone2007quantifying,stone2008effects}.

\subsection{Output levels are anticorrelated}

The ECF can be used to study the across-source modulation phenomenon
mathematically. Specifically, if the input envelopes $v_{s_{1}}[t,b]$
and $v_{s_{2}}[t,b]$ are statistically independent, then the covariance
between the output levels in each channel is negative:
\begin{equation}
\boxed{\mathrm{Cov}(v_{r_{1}}[t,b],v_{r_{2}}[t,b])\le0,\quad b=1,\dots,B.}\label{eq:output_cov}
\end{equation}

To prove this, we first show that the ECF is nondecreasing in one
envelope and nonincreasing in the other. Because lemmas
and theorems in this work follow from the properties of compression
functions, which do not depend directly on time or frequency, the
time and channel indices $[t,b]$ are omitted in their statements
and proofs.
\begin{lem}
\label{lem:ecf_nondecreasing}Any effective compression function $\hat{\mathcal{C}}(v_{1}|v_{2})$
is nondecreasing in $v_{1}$ and nonincreasing in $v_{2}$ for $v_{1},v_{2}\ge0$.
\end{lem}

\begin{proof}
Because $\mathcal{C}(v)$ is nondecreasing and $v_{2}$ is nonnegative,
$\hat{\mathcal{C}}(v_{1}|v_{2})=\mathcal{C}(v_{1}+v_{2})\frac{v_{1}}{v_{1}+v_{2}}$
is the product of two nondecreasing functions of $v_{1}$ and is therefore
nondecreasing. Because $\mathcal{C}(v)$ is concave and nonnegative,
$\mathcal{C}(v)/v$ is nonincreasing for $v>0$. Then $\mathcal{C}(v_{1}+v_{2})/(v_{1}+v_{2})$
is nonincreasing in $v_{2}$.
\end{proof}
Next, we will need the following result about functions of random
variables.
\begin{lem}
\label{lem:decreasing}If $f(x)$ is nondecreasing, $g(x)$ is nonincreasing,
$X$ is a random variable, and $\mathbb{E}[f(X)]$, $\mathbb{E}[g(X)]$,
and $\mathbb{E}[f(X)g(X)]$ exist, then 
\begin{equation}
\mathbb{E}\left[f(X)g(X)\right]\le\mathbb{E}[f(X)]\mathbb{E}[g(X)].
\end{equation}
\end{lem}

\begin{proof}
See Appendix \ref{sec:proof_decreasing}.
\end{proof}
We can now prove that independent envelopes become negatively correlated
when compressed.
\begin{thm}
\label{thm:cov}If $\hat{\mathcal{C}}(v_{1}|v_{2})$ is an effective
compression function and $V_{1}$ and $V_{2}$ are independent random
variables, then 
\begin{equation}
\mathrm{Cov}\left(\mathcal{\hat{C}}(V_{1}|V_{2}),\mathcal{\hat{C}}(V_{2}|V_{1})\right)\le0.
\end{equation}
\end{thm}

\begin{proof}
Because $\mathrm{Cov}\left(\hat{\mathcal{C}}(V_{1}|V_{2}),\hat{\mathcal{C}}(V_{2}|V_{1})\right)=\mathbb{E}\left[\hat{\mathcal{C}}(V_{1}|V_{2})\hat{\mathcal{C}}(V_{2}|V_{1})\right]-\mathbb{E}\left[\hat{\mathcal{C}}(V_{1}|V_{2})\right]\mathbb{E}\left[\hat{\mathcal{C}}(V_{2}|V_{1})\right]$,
it is sufficient to show that
\begin{equation}
\mathbb{E}\left[\hat{\mathcal{C}}(V_{1}|V_{2})\hat{\mathcal{C}}(V_{2}|V_{1})\right]\le\mathbb{E}\left[\hat{\mathcal{C}}(V_{1}|V_{2})\right]\mathbb{E}\left[\hat{\mathcal{C}}(V_{2}|V_{1})\right].
\end{equation}
From Lemma \ref{lem:ecf_nondecreasing}, $\hat{\mathcal{C}}(v_{1}|v_{2})$
is a nondecreasing function of $v_{1}$ and a nonincreasing function
of $v_{2}$. From iterated expectation and application of Lemma \ref{lem:decreasing},
we have
\begin{align}
 & \mathbb{E}\left[\hat{\mathcal{C}}(V_{1}|V_{2})\hat{\mathcal{C}}(V_{2}|V_{1})\right]\nonumber \\
 & =\mathbb{E}_{V_{2}}\left[\mathbb{E}_{V_{1}}\left[\hat{\mathcal{C}}(V_{1}|V_{2})\hat{\mathcal{C}}(V_{2}|V_{1})\mid V_{2}\right]\right]\\
 & \le\mathbb{E}_{V_{2}}\left[\mathbb{E}_{V_{1}}\left[\hat{\mathcal{C}}(V_{1}|V_{2})|V_{2}\right]\mathbb{E}_{V_{1}}\left[\hat{\mathcal{C}}(V_{2}|V_{1})|V_{2}\right]\right].
\end{align}
Now, because $V_{1}$ and $V_{2}$ are independent, $\mathbb{E}_{V_{1}}[\hat{\mathcal{C}}(V_{1}|V_{2})|V_{2}]$
is a nonincreasing function of $V_{2}$ and $\mathbb{E}_{V_{1}}[\hat{\mathcal{C}}(V_{2}|V_{1})|V_{2}]$
is a nondecreasing function of $V_{2}$. Applying Lemma \ref{lem:decreasing}
once more, 
\begin{align}
\mathbb{E}\left[\hat{\mathcal{C}}(V_{1}|V_{2})\hat{\mathcal{C}}(V_{2}|V_{1})\right] & \le\mathbb{E}_{V_{2}}\left[\mathbb{E}_{V_{1}}\left[\hat{\mathcal{C}}(V_{1}|V_{2})|V_{2}\right]\right]\nonumber \\
 & \quad\cdot\mathbb{E}_{V_{2}}\left[\mathbb{E}_{V_{1}}\left[\hat{\mathcal{C}}(V_{2}|V_{1})|V_{2}\right]\right]\\
 & =\mathbb{E}\left[\hat{\mathcal{C}}(V_{1}|V_{2})\right]\mathbb{E}\left[\hat{\mathcal{C}}(V_{2}|V_{1})\right].
\end{align}
\end{proof}
For linear gain, the theorem holds with equality because $\hat{\mathcal{C}}(v_{1}|v_{2})$
does not depend on $v_{2}$. The magnitude of the negative correlation
depends on the compression function; stronger compression causes the
ECFs and the conditional expectations to increase or decrease more
quickly, resulting in a stronger negative correlation. The channel
structure and time constants of the envelope detector affect the correlation
indirectly by altering the distributions of $V_{1}$ and $V_{2}$.

\subsection{Experiments}

\begin{figure}
\begin{centering}
\includegraphics{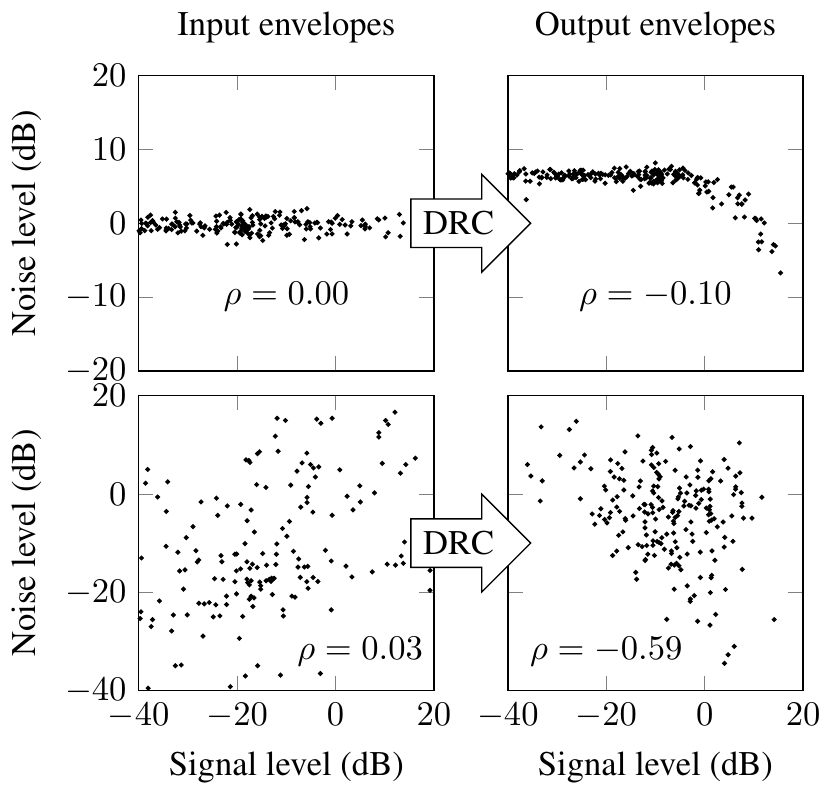}
\par\end{centering}
\caption{\label{fig:asmc}Sample input and output envelope pairs for mixtures
of two signals in a DRC system. Top: Speech and white noise. Bottom:
Speech and speech.}
\end{figure}

The ECF models the effects of level-dependent amplification for idealized
statistical envelopes, but does not perfectly represent the behavior
of real envelope detectors, especially their peak-tracking
dynamics. To verify the predicted result under more realistic conditions,
an experiment was conducted using mixtures of speech signals from
the VCTK database \citep{Veaux2017} and white Gaussian noise. All
signals have the same overall level so that the input SNR is 0 dB.
The DRC system has a constant compression ratio of 3:1, that is, a
cube-root compression function, in all channels; the filterbank uses
6 Mel-spaced bands from 0 to 8 kHz; and the envelope detector is the
filter from Eq. (\ref{eq:env_det}) with an attack
time of 10 ms and a release time of 50 ms as defined by ANSI S3.22-1996\nocite{ansi1996}
($\beta_{\text{a}}=0.978$ and $\beta_{\text{r}}=0.996$ at a sample
rate of 16 kHz). The input and output envelope samples were measured
using the same envelope detector applied to $s_{1}[t,b]$, $s_{2}[t,b]$,
$r_{1}[t,b]$, and $r_{2}[t,b]$.

Figure \ref{fig:asmc} illustrates the anticorrelation effects of
DRC. Each plot shows pairs of envelope samples for two signals in
a mixture: the input mixtures $(v_{s_{1}}[t,b],v_{s_{2}}[t,b])$ on
the left and the output mixtures $(v_{r_{1}}[t,b],v_{r_{2}}[t,b])$
on the right. The top plots are for speech and white noise and the
bottom plots are for two speech signals. The axes have been normalized
so that the average input level is 0 dB. The correlation coefficient
$\rho$ between samples is computed using the method of \citet{stone2007quantifying}
and averaged across the six channels. The input levels are mostly
uncorrelated between the two component signals, but the DRC system
shifts the levels according to a vector field similar to that of Fig.
\ref{fig:quiver}, producing correlated output points. Because the
white noise has nearly constant envelope, the effect of DRC is most
apparent at high speech signal levels. When the speech signal is strong,
both speech and noise are attenuated, bending the distribution of
level pairs downward and producing a negative correlation. The DRC
system has a similar effect on the speech mixture; each signal's level
decreases as the level of the other increases, producing a large negative
correlation. In an experiment with similar compression parameters
but a higher input SNR, \citet{stone2008effects} found an output
correlation coefficient of $-0.11$ between speech signals.

It has been shown that the correlation effect is weaker for slower-acting
compression, which averages signal levels over hundreds of milliseconds
and therefore does not vary gain levels as quickly \citep{stone2007quantifying,alexander2015effects}.
In a slow-acting compressor, the envelope samples ($v_{s_{1}}[t,b],v_{s_{2}}[t,b]$)
are not as widely spread in the $(v_{1},v_{2})$ plane and so they
are less distorted by the gain field illustrated in Fig. \ref{fig:quiver}.

\section{\label{sec:effective_compression}Effective Compression Performance}

The interaction between signals in a DRC system not only introduces
distortion to the signal envelopes, it also makes the compressor itself
less effective. Intuitively, if a signal of interest is weaker than
a noise source, then the noise level will determine the gain applied
to both signals and the target signal will not be compressed. Even
when the target signal has a higher level, the noise will cause the
gain to decrease more slowly than it should with respect to the target
level.

To quantify the effect of noise on compression performance, we can
measure the change in the output level of the target signal in response
to a change in its input level and compare that relationship to the
nominal compression ratio. In the hearing literature, it has been
observed experimentally that the effective compression ratio of a
DRC system is reduced in the presence of noise \citep{souza2006measuring}.
In this section we prove the equivalent result
that noise increases the effective compression slope, defined as the
log-log slope of the ECF.
\begin{defn}
\label{def:ecs}If $\hat{\mathcal{C}}_{b}(v_{1}|v_{2})$ is differentiable
with respect to $v_{1}$, then the \emph{effective compression slope}
$\hat{\mathrm{CS}}_{b}(v_{1}|v_{2})$ is given by 
\begin{align}
\hat{\mathrm{CS}}_{b}(v_{1}|v_{2}) & =\frac{\partial}{\partial u}\hat{\mathcal{C}}_{b}(e^{u}|v_{2})|_{u=\ln v_{1}}\\
 & =\frac{\frac{\partial}{\partial v_{1}}\hat{\mathcal{C}}_{b}(v_{1}|v_{2})}{\hat{\mathcal{C}}_{b}(v_{1}|v_{2})}v_{1}.
\end{align}
\end{defn}

\subsection{Noise reduces compression performance}

Using the properties of the ECF, it can be shown that the effective
compression slope from Definition \ref{def:ecs} is larger than the
nominal compression slope from Definition \ref{def:cs}---equivalently,
the effective compression ratio is smaller than the nominal compression
ratio---meaning that in noise the system is less compressive on each
component signal than it would be without noise:
\begin{equation}
\boxed{\hat{\mathrm{CS}}_{b}(v_{s_{1}}|v_{s_{2}})\ge\mathrm{CS}_{b}(v_{s_{1}}+v_{s_{2}}),\quad b=1,\dots,B.}
\end{equation}
The result follows from the concavity of the ECF.
\begin{thm}
\label{thm:ecs_comp}If a compression function $\mathcal{C}(v)$ is
differentiable at $v_{x}=v_{1}+v_{2}$, then its effective compression
slope satisfies
\begin{equation}
\hat{\mathrm{CS}}(v_{1}|v_{2})\ge\mathrm{CS}(v_{x}),
\end{equation}
with equality if $\mathcal{C}(v)$ is linear or if $v_{2}=0$.
\end{thm}

\begin{proof}
Because $\mathcal{C}(v)$ is defined to be concave and nonnegative
for $v>0$, it follows that 
\begin{equation}
\mathcal{C}(v)-v\mathcal{C}'(v)\ge0
\end{equation}
for all $v$ at which $\mathcal{C}$ is differentiable, with equality
if $\mathcal{C}$ is linear. The effective compression slope is given
by
\begin{align}
\hat{\mathrm{CS}}(v_{1}|v_{2}) & =\frac{\frac{\partial}{\partial v_{1}}\hat{\mathcal{C}}(v_{1}|v_{2})}{\hat{\mathcal{C}}(v_{1}|v_{2})}v_{1}\\
 & =\frac{v_{x}}{\mathcal{C}(v_{x})}\left(\frac{\mathcal{C}'(v_{x})v_{1}+\mathcal{C}(v_{x})}{v_{x}}-\frac{\mathcal{C}(v_{x})v_{1}}{v_{x}^{2}}\right)\\
 & =\frac{\mathcal{C}'(v_{x})}{\mathcal{C}(v_{x})}v_{1}+1-\frac{v_{1}}{v_{x}}\\
 & =\frac{\mathcal{C}'(v_{x})}{\mathcal{C}(v_{x})}v_{x}-\frac{\mathcal{C}'(v_{x})}{\mathcal{C}(v_{x})}v_{2}+\frac{v_{2}}{v_{x}}\\
 & =\mathrm{CS}(v_{x})+\frac{v_{2}}{v_{x}\mathcal{C}(v_{x})}(\mathcal{C}(v_{x})-v_{x}\mathcal{C}'(v_{x}))\label{eq:cs_proof_line}\\
 & \ge\mathrm{CS}(v_{x})
\end{align}
with equality if $\mathcal{C}$ is linear or if $v_{2}=0$.
\end{proof}
Equation (\ref{eq:cs_proof_line}) illustrates that the effective
compression slope for a target signal increases with the level of
the interfering signal. In fact, in the limit as $v_{s_{1}}/v_{x}$
approaches 0, the effective compression slope approaches 1, so that
the system applies linear gain to the target signal. In the low-SNR
regime, the gain applied to both signals is determined by the level
of the interfering signal. The theorem shows, however, that even at
high SNR the compression effect is slightly reduced.

\subsection{Experiments}

\begin{figure}
\begin{centering}
\includegraphics{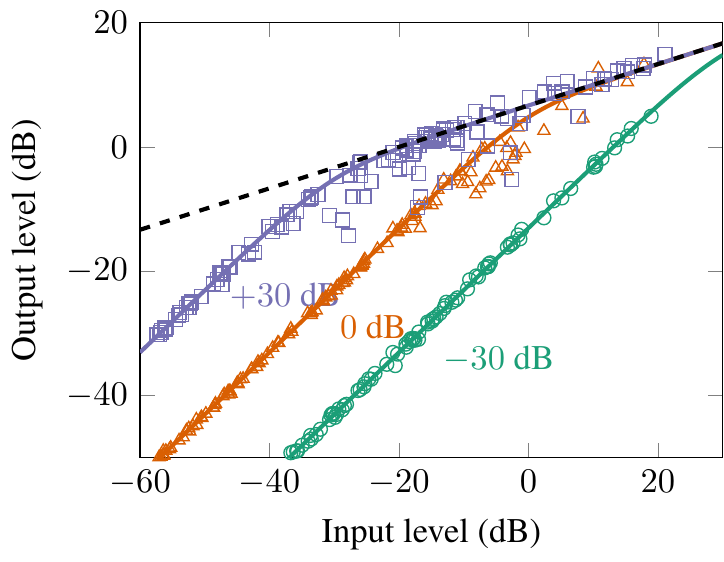}
\par\end{centering}
\caption{\label{fig:ecs}Effective compression function for speech and white
noise at different SNRs. The dashed line shows the nominal compression
function, the curves show the predicted effective compression functions,
and the plotted points show envelope samples measured from the simulated
DRC system.}
\end{figure}

Theorem \ref{thm:ecs_comp} shows that under an idealized model, additive
noise always reduces the effect of compression on a signal of interest.
The model, however, does not account for correlations between signals,
for the temporal dynamics of the envelope detector, or for peak-tracking
effects. To verify this result experimentally, a DRC system was simulated
using the same parameters as in the previous section. The signal of
interest is a speech recording and the interfering signal is white noise
with variable level.

Figure \ref{fig:ecs} shows the effective compression performance
of the system for three signal-to-noise ratios. The dashed
line shows the nominal compression function $\mathcal{C}_{b}(v)=v^{1/3}$
for all $b$. The solid curves are the theoretical effective compression
functions $\hat{\mathcal{C}}_{b}(v_{s_{1}}|v_{s_{2}})$ for constant
noise power $v_{s_{2}}$. The plotted points are samples of the measured
input and output envelopes across all channels. Each curve is nearly linear
at speech levels much lower than the
noise level and closely matches the nominal compression curve at speech
levels much higher than the noise level. At high SNR, the speech is
compressed correctly, while at low SNR, the gain is determined by
the noise level.

These results are consistent with previous work. \citet{souza2006measuring}
measured the effective compression ratio (ECR) empirically using a
simulated DRC system and the phase inversion method. The ECR was calculated
as the dynamic range of the input envelope divided by the dynamic
range of the output envelope, where the dynamic range is defined as
the difference between the 95th and 5th percentiles of sample levels
in the signals. The ratio was found to increase monotonically with
SNR, from 1.06 at $-2$ dB to $1.12$ in quiet for a nominal compression
ratio of 2:1. Using the method of \citet{souza2006measuring} and
averaging across signal bands, the ECRs from the experiments here
were 1.00 at $-30$ dB SNR, 1.14 at $0$ dB, and 1.69 at $+30$ dB
with a nominal compression ratio of 3:1.

\section{\label{sec:snr}Signal-to-Noise Ratio}

Of the three distortion effects discussed in this work, the impact
of DRC on long-term signal-to-noise ratio is both the most studied
empirically and the most challenging to analyze mathematically. The
levels of the signal components in the output of a compression system
depend on the input levels and temporal dynamics of both the signal
of interest and the noise. Using the phase inversion method,
\citet{souza2006measuring} found that simulated DRC processing reduced
the SNR of a speech signal in speech-shaped noise while linear processing
did not. \citet{naylor2009long} observed this effect in commercial
hearing aids and showed that it depends on the type of noise,
filterbank structure, and envelope time constants. \citet{brons2015acoustical}
and \citet{miller2017output} demonstrated the effect in more recent
hearing aids that include noise reduction algorithms, which are also
nonlinear and can interact with DRC in complex ways \citep{kortlang2018evaluation}.
Like the other nonlinear distortion effects described here, SNR reduction
appears to be more severe for fast compression \citep{alexander2015effects,may2018signal}.
\citet{reinhart2017effects} found that reverberant signals are less
affected. Listening tests suggest that fast-acting compression can
improve intelligibility with some types of noise but not others, and
that these effects depend on the input SNR \citep{rhebergen2009dynamic,rhebergen2017characterizing,kowalewski2018effects}. 

\subsection{Signal-to-noise ratio in constant-envelope noise}

While it is difficult to say much in general about the effect of compression
on SNR, we can prove a result for an important special case: a target
signal with a time-varying envelope and a noise signal with constant
envelope. Information-rich signals such as speech tend to vary rapidly
with time. Many classic speech enhancement algorithms, such as spectral
subtraction, rely on the assumption that the noise spectrum is constant
while the spectrum of the speech signal varies \citep{loizou2013speech}.
At times and frequencies where the input level is large, it is assumed
that speech is present and the signal is amplified, while at lower
levels the signal is attenuated to remove noise. Because these speech
enhancement systems amplify high-level signals and attenuate low-level
signals, they act as dynamic range \emph{expanders}.

If a dynamic range expander can improve SNR, it stands to reason that
a compressor might make it worse. To see why, let us analyze the effect
of compression on the average SNR over time. Because the envelope
is proportional to the power of a signal component, the average SNR
at the input is given by
\begin{equation}
\mathrm{SNR}_{\mathrm{in}}[b]=\frac{\mathrm{mean}_{t}v_{s_{1}}[t,b]}{\mathrm{mean}_{t}v_{s_{2}}[t,b]},\quad b=1,\dots,B,
\end{equation}
and the average SNR at the output is 
\begin{align}
\mathrm{SNR}_{\mathrm{out}}[b] & =\frac{\mathrm{mean}_{t}v_{r_{1}}[t,b]}{\mathrm{mean}_{t}v_{r_{2}}[t,b]}\\
 & =\frac{\mathrm{mean}_{t}\hat{\mathcal{C}}_{b}(v_{s_{1}}[t,b]|v_{s_{2}}[t,b])}{\mathrm{mean}_{t}\hat{\mathcal{C}}_{b}(v_{s_{2}}[t,b]|v_{s_{1}}[t,b])}.
\end{align}

If the compression function were linear, then the input and output
SNRs would be identical. For a concave compression function with convex
gain, it can be shown that, if the noise envelope is constant, then
the output SNR is lower than the input SNR:
\begin{equation}
\boxed{\mathrm{SNR}_{\mathrm{out}}[b]\le\mathrm{SNR}_{\mathrm{in}}[b],\quad b=1,\dots,B.}
\end{equation}

The proofs in this section rely on an additional technical condition
on the compression function. Not only must $\mathcal{C}_{b}(v)$ be
nonnegative and concave, the gain function $\mathcal{C}_{b}(v)/v$
must be convex. This condition is satisfied for many smooth compression
functions, including linear, power-law, and logarithmic, but not for
some functions with corners like that in Fig. \ref{fig:Compression-functions}.
This condition ensures that the effective compression function is
concave in its first argument and convex in its second.
\begin{lem}
\label{lem:convex}If $\mathcal{C}(v)$ is a compression function
and $\mathcal{C}(v)/v$ is convex for all $v>0$, then the effective
compression function $\hat{\mathcal{C}}(v_{1}|v_{2})$ is concave
in $v_{1}$ and convex in $v_{2}$.
\end{lem}

\begin{proof}
See Appendix \ref{sec:proof_convex}.
\end{proof}
This property, along with the condition that the interference signal
has constant envelope, allows us to prove that the output SNR is no
larger than the input SNR.
\begin{thm}
\label{thm:snr}If $\mathcal{C}(v)$ is a compression function and
$\mathcal{C}(v)/v$ is convex for all $v>0$, $v_{1}[t]>0$ for all
$t$, and $v_{2}[t]=\bar{v}_{2}>0$ for all $t$, then
\begin{equation}
\mathrm{SNR}_{\mathrm{out}}\le\mathrm{SNR}_{\mathrm{in}}
\end{equation}
with equality if $v_{1}[t]$ is constant or $\mathcal{C}$ is linear.
\end{thm}

\begin{proof}
Since $v_{2}[t]$ is fixed, the output SNR can be written 
\begin{equation}
\mathrm{SNR}_{\mathrm{out}}=\frac{\mathrm{mean}_{t}\hat{\mathcal{C}}(v_{1}[t]|\bar{v}_{2})}{\mathrm{mean}_{t}\hat{\mathcal{C}}(\bar{v}_{2}|v_{1}[t])}.\label{eq:snr_out}
\end{equation}
The numerator is the mean over $t$ of a concave function of $v_{1}[t]$.
By Jensen's inequality \citep{cover2006elements}, 
\begin{equation}
\mathrm{mean}_{t}\hat{\mathcal{C}}(v_{1}[t]|\bar{v}_{2})\le\hat{\mathcal{C}}(\mathrm{mean}_{t}v_{1}[t]|\bar{v}_{2}),
\end{equation}
with equality when $\mathcal{C}$ is linear or $v_{1}[t]$ is constant.
Similarly, the denominator is the mean over $t$ of a convex function
of $v_{1}[t]$. Again applying Jensen's inequality,
\begin{equation}
\mathrm{mean}_{t}\hat{\mathcal{C}}(\bar{v}_{2}|v_{1}[t])\ge\hat{\mathcal{C}}(\bar{v}_{2}|\mathrm{mean}_{t}v_{1}[t]),
\end{equation}
with equality when $\mathcal{C}$ is linear or $v_{1}[t]$ is constant.
Let $\bar{v}_{1}=\mathrm{mean}_{t}v_{1}[t]$. Since the numerator
and denominator of (\ref{eq:snr_out}) are both positive, we have
\begin{align}
\mathrm{SNR}_{\mathrm{out}} & \le\frac{\hat{\mathcal{C}}(\bar{v}_{1}|\bar{v}_{2})}{\hat{\mathcal{C}}(\bar{v}_{2}|\bar{v}_{1})}\\
 & =\frac{\bar{v}_{1}\mathcal{C}(\bar{v}_{1}+\bar{v}_{2})/(\bar{v}_{1}+\bar{v}_{2})}{\bar{v}_{2}\mathcal{C}(\bar{v}_{1}+\bar{v}_{2})/(\bar{v}_{1}+\bar{v}_{2})}\\
 & =\frac{\bar{v}_{1}}{\bar{v}_{2}}\\
 & =\mathrm{SNR}_{\mathrm{in}}
\end{align}
with equality when $\mathcal{C}$ is linear or $v_{1}[t]$ is constant.
\end{proof}

\subsection{Experiments}

The equality condition in Theorem \ref{thm:snr} suggests that the
SNR-reducing effect of compression depends on the curvature of the
compression function. Figure \ref{fig:snr_ratio} compares the input
and output SNRs for speech in white noise at different compression
ratios. The simulations used a knee-shaped compression function like
that in Fig. \ref{fig:Compression-functions}. Although this compressor
does not meet the technical condition required
for Lemma \ref{lem:convex}, the results still show the behavior predicted
by the theorem. The SNR-reducing effect is greatest at high input
SNRs; at low SNRs, the noise level determines the gain and the effective
compression function is linear, so the SNR is not affected. Furthermore,
higher compression ratios have stronger effects. This phenomenon has
been observed in the hearing literature as well: \citet{rhebergen2009dynamic}
found similar input-output SNR curves for speech in stationary and
interrupted noise (see Fig. 8 of that work). \citet{naylor2009long}
obtained similar results using commercial hearing aids configured
with different compression ratios (see Fig. 4 of that work). 

\begin{figure}
\begin{centering}
\includegraphics{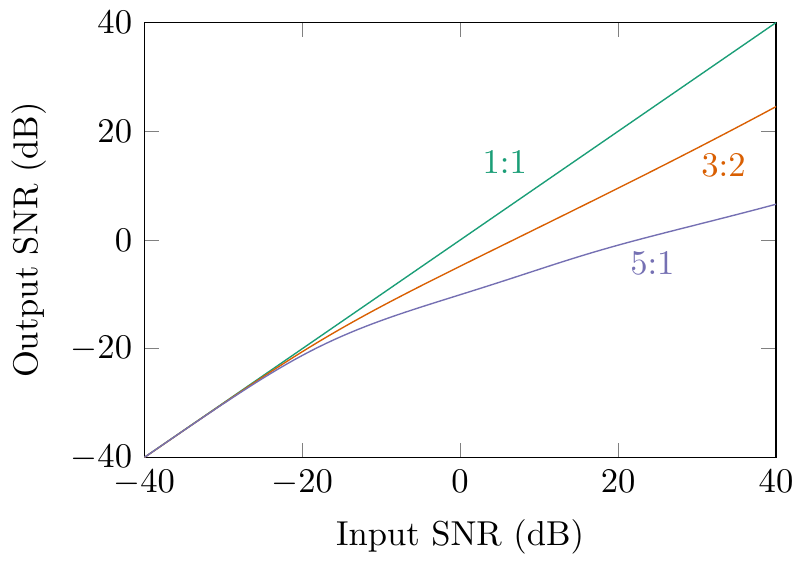}
\par\end{centering}
\caption{\label{fig:snr_ratio}Effect of DRC with different compression ratios
on long-term SNR of speech in white noise.}
\end{figure}

Theorem \ref{thm:snr} applies only to constant-envelope noise. While
this is an important special case, it does not reflect most real-world
sound mixtures. When the target and noise signals both vary with time,
the compressor reduces the dynamic range of both signals. The weaker
signal will be amplified more and the stronger signal less, pushing
their average output levels closer together.
Figure \ref{fig:noise_types} shows the results of the SNR experiment
with 3:1 knee-shaped compression and different noise types. With white
noise, the SNR is always reduced, as predicted by Theorem \ref{thm:snr}.
With speech babble, generated by mixing fourteen VCTK speech
clips, the SNR is slightly increased at low input SNRs. When the target
and interference signals are both single-talker speech signals, the
SNR is improved when it is negative but made worse when it is positive.

\begin{figure}
\begin{centering}
\includegraphics{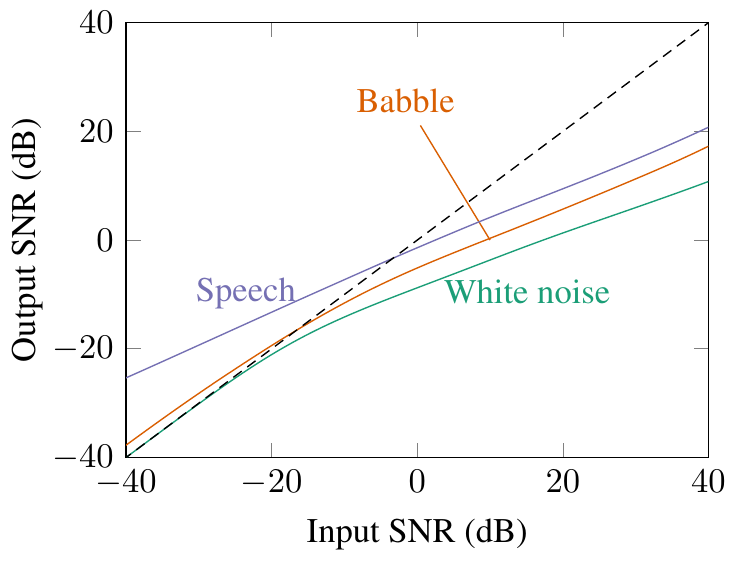}
\par\end{centering}
\caption{\label{fig:noise_types}Effect of compression on long-term signal-to-noise
ratio of mixtures of speech with different types of noise}
\end{figure}

These results align well with those observed in the hearing literature.
\citet{naylor2009long} found that output SNR is always reduced for
speech in unmodulated noise, greatly reduced at positive input SNR
and slightly reduced at negative input SNR for speech in modulated
noise, and symmetrically increased at negative input SNR and decreased
at positive input SNR for a mixture of two speech signals (see Fig.
3 of that work). \citet{reinhart2017effects} performed experiments
with different numbers of talkers and found a similar symmetric relationship
for a mixture of two talkers. The SNR improvement at negative input
SNRs declined with each additional interfering talker (see Fig. 3
of that work), consistent with the results for speech babble here.

\section{Conclusions}

The mathematical analysis above confirms the empirical evidence from
the hearing literature that DRC causes distortion in noise. The effects
of this distortion depend on the characteristics of the signals, especially
their relative levels. The across-source modulation effect is most
pronounced at SNRs near unity: when two signals are similar in level,
they modulate each other, resulting in a negative correlation between
their envelopes. At low SNR, the effective compression function for
the target signal becomes nearly linear and the dynamic range of that
signal is not changed. Compression algorithms are thus ineffective in 
challenging listening conditions where they would
presumably help most. Meanwhile, at high SNR, the signal of interest
is amplified by less than the noise, reducing average SNR.

Can anything be done to improve the performance of DRC systems in
noise? The analysis shows that all these effects are caused by the
concave curvature of the compression function, which is also what
makes the system compressive. It seems, then, that distortion
is inevitable whenever signals are compressed as a mixture. 

A possible solution is to compress the component signals of a mixture
independently, as music producers do when mixing instrumental
and vocal recordings. Listening tests have shown improved intelligibility
when signals are compressed before rather than after mixing \citep{stone2008effects,rhebergen2009dynamic}.
Of course, real hearing aids do not have access to the unmixed
source signals, so a practical multisource compression system must
perform source separation. \citet{hassager2017preserving} used a
single-microphone classification method to separate direct from reverberant
signal components, helping to preserve spatial cues that can be distorted
by DRC. \citet{may2018signal} proposed a single-microphone separation
system that applies fast-acting compression to speech components and
slow-acting compression to noise components; listening experiments
with an oracle separation algorithm improved both quality
and intelligibility \citep{kowalewski2020perceptual}. \citet{corey2017compression}
used a multimicrophone separation method to apply separate compression functions to
each of several competing speech signals. The output exhibited better
objective measures of across-source modulation distortion, effective
compression performance, and SNR compared to a conventional system.
It was later shown that larger wearable microphone arrays can provide
better multisource compression performance than the small arrays contained
in hearing aid earpieces \citep{corey2019thesis}. 

The mathematical tools introduced in this work can help
researchers to understand the distortion effects of conventional DRC
systems in noise and to devise new approaches to nonlinear processing
for mixtures of multiple signals. The effective compression function
models interactions between signal envelopes at the input and output
of a DRC system. However, it does not give insights about
the effects of the envelope detector or filterbank, which are known
to affect the magnitude of all three distortion effects. Further analysis 
could explain how channel structure and attack and release
times affect the distribution of envelope samples.

Like the human auditory system itself, dynamic range compression is
a complex nonlinear system that defies simple analysis. By
modeling how DRC systems work in the presence of noise, we can design
listening systems to help people hear better even in the most challenging
environments.

\appendix

\section{\label{sec:proof_decreasing}Proof of Lemma 2}
\begin{manualtheorem}{2}
If $f(x)$ is nondecreasing, $g(x)$ is nonincreasing, $X$ is a random
variable, and $\mathbb{E}[f(X)]$, $\mathbb{E}[g(X)]$, and $\mathbb{E}[f(X)g(X)]$
exist, then 
\begin{equation}
\mathbb{E}\left[f(X)g(X)\right]\le\mathbb{E}[f(X)]\mathbb{E}[g(X)].
\end{equation}
\end{manualtheorem}

\begin{proof}
Because $f(x)$ is nondecreasing and $g(x)$ is nonincreasing, for
every $x$ and $y$ we have 
\begin{equation}
\left[f(x)-f(y)\right]\left[g(x)-g(y)\right]\le0.
\end{equation}
It is sufficient to show that $\mathbb{E}[f(X)g(X)]-\mathbb{E}[f(X)]\mathbb{E}[g(X)]\le0$.
If $X$ has cumulative distribution function $P(x)$, then 
\begin{align}
 & \mathbb{E}[f(X)g(X)]-\mathbb{E}[f(X)]\mathbb{E}[g(X)]\nonumber \\
 & =\int_{x}f(x)g(x)\,\mathrm{d}P(x)-\int_{x}f(x)\,\mathrm{d}P(x)\int_{y}g(y)\,\mathrm{d}P(y)\\
 & =\int_{x}\int_{y}f(x)[g(x)-g(y)]\,\mathrm{d}P(y)\mathrm{d}P(x)\\
 & =\int_{x}\int_{y<x}f(x)[g(x)-g(y)]\,\mathrm{d}P(y)\mathrm{d}P(x)\nonumber \\
 & \quad+\int_{y}\int_{x<y}f(x)[g(x)-g(y)]\,\mathrm{d}P(x)\mathrm{d}P(y)\label{eq:fubini}\\
 & =\int_{x}\int_{y<x}f(x)[g(x)-g(y)]\,\mathrm{d}P(y)\mathrm{d}P(x)\nonumber \\
 & \quad+\int_{x}\int_{y<x}f(y)[g(y)-g(x)]\,\mathrm{d}P(y)\mathrm{d}P(x)\label{eq:variable_swap}\\
 & =\int_{x}\int_{y<x}[f(x)-f(y)][g(x)-g(y)]\,\mathrm{d}P(y)\mathrm{d}P(x)\\
 & \le0.
\end{align}
Line (\ref{eq:fubini}) swaps the order of integration using Fubini's
theorem \citep{knapp2005real} and line (\ref{eq:variable_swap}) exchanges the integration
variables $x$ and $y$.
\end{proof}

\section{\label{sec:proof_convex}Proof of Lemma 3}
\begin{manualtheorem}{3}
If $\mathcal{C}(v)$ is a compression function and $\mathcal{C}(v)/v$
is convex for all $v>0$, then the effective compression function
$\hat{\mathcal{C}}(v_{1}|v_{2})$ is concave in $v_{1}$ and convex
in $v_{2}$.
\end{manualtheorem}

\begin{proof}
Starting with Definition \ref{def:ecf} and letting $v_{1}=\lambda p+(1-\lambda)q$,
\begin{align}
\hat{\mathcal{C}}(v_{1}|v_{2}) & =\frac{\mathcal{C}(\lambda p+(1-\lambda)q+v_{2})}{\lambda p+(1-\lambda)q+v_{2}}(\lambda p+(1-\lambda)q)\\
 & =\mathcal{C}(\lambda(p+v_{2})+(1-\lambda)(q+v_{2}))\\
 & \quad-v_{2}\frac{\mathcal{C}(\lambda(p+v_{2})+(1-\lambda)(q+v_{2}))}{\lambda(p+v_{2})+(1-\lambda)(q+v_{2})}.
\end{align}
Because $\mathcal{C}(v)$ is concave and $\mathcal{C}(v)/v$ is convex,
\begin{align}
\hat{\mathcal{C}}(v_{1}|v_{2}) & \ge\lambda\mathcal{C}(p+v_{2})+(1-\lambda)\mathcal{C}(q+v_{2})\\
 & \quad-v_{2}\left(\lambda\frac{\mathcal{C}(p+v_{2})}{p+v_{2}}+(1-\lambda)\frac{\mathcal{C}(q+v_{2})}{q+v_{2}}\right)\\
 & =\lambda\frac{\mathcal{C}(p+v_{2})}{p+v_{2}}p+(1-\lambda)\frac{\mathcal{C}(q+v_{2})}{q+v_{2}}\\
 & =\lambda\hat{\mathcal{C}}(p|v_{2})+(1-\lambda)\hat{\mathcal{C}}(q|v_{2}).
\end{align}
Therefore $\hat{\mathcal{C}}(v_{1}|v_{2})$ is concave in $v_{1}$.

Similarly, letting $v_{2}=\lambda p+(1-\lambda)q$,
\begin{align}
\hat{\mathcal{C}}(v_{1}|v_{2}) & =\frac{\mathcal{C}(v_{1}+\lambda p+(1-\lambda)q)}{v_{1}+\lambda p+(1-\lambda)q}v_{1}\\
 & =\frac{\mathcal{C}(\lambda(v_{1}+p)+(1-\lambda)(v_{1}+q))}{\lambda(v_{1}+p)+(1-\lambda)(v_{1}+q)}v_{1}\\
 & \le\lambda\frac{\mathcal{C}(v_{1}+p)}{v_{1}+p}v_{1}+(1-\lambda)\frac{\mathcal{C}(v_{1}+q)}{v_{1}+q}v_{1}\\
 & =\lambda\hat{\mathcal{C}}(v_{1}|p)+(1-\lambda)\hat{\mathcal{C}}(v_{1}|q).
\end{align}
Therefore $\hat{\mathcal{C}}(v_{1}|v_{2})$ is convex in $v_{2}$.
\end{proof}

\end{document}